\documentclass[11pt]{article}
\usepackage[numbers]{natbib}
\usepackage{fullpage}
\usepackage{url,ifthen}
\usepackage{srcltx}
\usepackage{multirow}
\usepackage{boxedminipage}
\usepackage[margin=1.1in]{geometry}
\usepackage{nicefrac}
\usepackage{xspace}
\usepackage{graphicx}
\usepackage{srcltx}
\usepackage[usenames]{color}
\usepackage{xspace}
\definecolor{DarkGreen}{rgb}{0.1,0.5,0.1}
\definecolor{DarkRed}{rgb}{0.5,0.1,0.1}
\definecolor{DarkBlue}{rgb}{0.1,0.1,0.5}
\usepackage[small]{caption}
\usepackage{float}
\usepackage{balance}
\usepackage{amsmath}
\usepackage{amsfonts}
\usepackage{amssymb}
\usepackage{amsthm}

\newtheorem{theorem}{Theorem}[section]
\newtheorem{lemma}[theorem]{Lemma}
\newtheorem{claim}[theorem]{Claim}
\newtheorem{corollary}[theorem]{Corollary}
\theoremstyle{definition}
\newtheorem{definition}[theorem]{Definition}
\usepackage{hyperref}
\newcommand{\ignore}[1]{}

\renewcommand{\Pr}{\mathop{\bf Pr\/}}                    
\newcommand{\E}{\mathop{\bf E\/}}

\newcommand{\poly}{\mathrm{poly}}

\newcommand{\R}{\mathbb R}

\makeatletter
\floatstyle{ruled}
\newfloat{fragment}{H}{lop}
\floatname{fragment}{Algorithm}
\renewcommand{\floatc@ruled}[2]{\vspace{2pt}{\@fs@cfont \#1.\:} \#2 \par
 \vspace{1pt}}
\makeatother

\renewenvironment{proof}{\medskip\noindent{\textbf{Proof:}}} {$\blacksquare$\vskip \belowdisplayskip}

\title{New Algorithms for Learning Incoherent and Overcomplete Dictionaries }

 \author{Sanjeev Arora \thanks{arora@cs.princeton.edu, Princeton University, Computer Science Department and Center for Computational Intractability}
\and
Rong Ge \thanks{rongge@microsoft.com, Microsoft Research, New England.Part of this work was done while the author was a graduate student at Princeton University and was supported in part by NSF grants CCF-0832797, CCF-1117309,
CCF-1302518, DMS-1317308, and Simons Investigator}
\and
Ankur Moitra \thanks{moitra@mit.edu, Massachusetts Institute of Technology, Department of Mathematics and CSAIL Part of this work was done while the author was a postdoc at the Institute for Advanced Study and was supported in part 
by NSF grant No.
DMS-0835373 and by an NSF Computing and Innovation Fellowship.}
 }

\begin{document}

\maketitle

\begin{abstract}
In {\em sparse recovery} we are given a matrix $A \in \R^{n\times m}$ (\textquotedblleft the dictionary\textquotedblright) and a vector of the form $A X$ where $X$ is {\em sparse}, and the goal is to recover $X$. This is a central notion in signal processing, statistics and machine learning. But in applications such as {\em sparse coding}, edge detection, compression and super resolution, the dictionary $A$ is unknown and has to be learned from random examples of the form $Y = AX$ where $X$ is drawn from an appropriate distribution --- this is the {\em dictionary learning} problem. In most settings, $A$ is {\em overcomplete}: it has more columns than rows. This paper presents a polynomial-time algorithm for learning overcomplete dictionaries; the only previously known algorithm with provable guarantees is the recent work of \cite{SWW} who gave an algorithm for the undercomplete case, which is rarely the case in applications. Our algorithm applies to {\em incoherent} dictionaries which have been a central object of study since they were introduced in seminal work of \cite{DH}. In particular, a dictionary is $\mu$-incoherent if each pair of columns has inner product at most $\mu / \sqrt{n}$.

The algorithm makes natural stochastic assumptions about the unknown sparse vector $X$, which can contain $k \leq c \min(\sqrt{n}/\mu \log n, m^{1/2 - \eta})$ non-zero entries (for any $\eta > 0$). This is close to the best $k$ allowable by the best sparse recovery algorithms  {\em even if one knows the dictionary $A$ exactly}. Moreover, both the running time and sample complexity depend on $\log 1/\epsilon$, where $\epsilon$ is the target accuracy, and so our algorithms converge very quickly to the true dictionary. Our algorithm can also tolerate substantial amounts of noise provided it is incoherent with respect to the dictionary (e.g., Gaussian). In the noisy setting, our running time and sample complexity depend polynomially on $1/\epsilon$, and this is necessary. 

\end{abstract}

\section{Introduction}


Finding {\em sparse} representations for data ---signals, images, natural language--- is a major focus of computational harmonic analysis~\citep{E,M}. 
This requires having the right dictionary $A \in \R^{n \times m}$ for the dataset, which allows each data point to be written as a sparse linear combination of the columns of $A$.
For images, popular choices for the dictionary include sinusoids, wavelets, ridgelets, curvelets, etc.~\citep{M} and each one is useful for different types of features: wavelets for impulsive events, ridgelets for discontinuities in edges, curvelets for smooth curves, etc. It is common to combine such hand-designed bases into a single dictionary, which is ``redundant" or ``overcomplete" because $m \gg n$. 
This can allow sparse representation even if an image contains many different ``types" of features jumbled together. 
In machine learning dictionaries are also used for {\em feature selection}~\citep{PAE} and for building classifiers on top of sparse coding primitives~\citep{KRL}. 

In many settings hand-designed dictionaries do not do as well as dictionaries that are fit to the dataset using automated methods.
In image processing such discovered dictionaries are used to perform denoising \citep{EA}, edge-detection \citep{MLBHP}, super-resolution \citep{YWHM} and compression. 
The problem of discovering the best dictionary to a dataset is called {\em dictionary learning} and also referred to as {\em sparse coding} in machine learning. 
Dictionary learning is also a basic building block  in the design of deep learning systems~\citep{RBL}. See~\cite{A, E} for further applications. 
In fact, the dictionary learning problem  was identified by  \cite{OF} as part of a study on 
internal image representations in the visual cortex. Their work suggested that basis vectors in learned dictionaries often correspond to  
well-known image filters such as Gabor filters.

Our goal is to design an algorithm for this problem with provable guarantees in the same spirit as recent work on nonnegative matrix factorization \citep{AGKM}, topic models \citep{AGM,AFHKL} and mixtures models \citep{MV,BS}. (We will later discuss why current algorithms in~\cite{LS}, \cite{EAH}, \cite{AEB}, \cite{KMRELS}, \cite{LBRN} do not come with such guarantees.) 
Designing such algorithms for dictionary learning has proved challenging. Even if the dictionary is completely known, it can be NP-hard to represent a vector $u$ as a sparse linear combination of the columns of $A$~\citep{DMA}. 
However for many natural types of dictionaries, the problem of finding a sparse representation is computationally easy. The pioneering work of \cite{DH}, \cite{DE} and \cite{GN} (building on the uncertainty principle of \cite{DS}) presented a number of important examples (in fact, the ones we used above) of dictionaries that are {\em incoherent} and showed that $\ell_1$-minimization can find a sparse representation in a known, incoherent dictionary if one exists.

\begin{definition}[$\mu$-incoherent]
An $n\times m$ matrix $A$ whose columns are {\em unit} vectors is  {\em $\mu$-incoherent} if  $\forall i \ne j$ we have $\langle A_i, A_j \rangle \leq \mu/\sqrt{n}.$
We will refer to $A$ as {\em incoherent} if $\mu$ is $O(\log n)$.
\end{definition}

\noindent A randomly chosen dictionary is incoherent with high probability (even if $m = n^{100}$). \cite{DH} gave many other important examples of incoherent dictionaries, such as one constructed from {\em spikes} and {\em sines}, as well as those built up from wavelets and sines, or even wavelets and ridgelets. There is a rich body of literature devoted to incoherent dictionaries (see additional references in \cite{GMS}). \cite{DH} proved that given $u = Av$ where $v$ has $k$ nonzero entries, where $k \leq \sqrt{n}/2\mu$, {\em basis pursuit} (solvable by a linear program) recovers $v$ exactly and it is unique. \cite{GMS} (and subsequently \cite{TGMS}) gave algorithms for recovering $v$ even in the presence of additive noise. \cite{T} gave a more general {\em exact recovery condition} (ERC) under which the sparse recovery problem for incoherent dictionaries can be algorithmically solved. All of these require $n > k^2 \mu^2$. In a foundational work, \cite{CRT} showed that basis pursuit solves the sparse recovery problem even for $n = O(k \log (m/k))$ if $A$ satisfies the weaker {\em restricted isometry property}~\citep{CT}. 
Also if $A$ is a full-rank square matrix, then we can compute $v$ from $A^{-1}u$, trivially. But our focus here will be on incoherent and overcomplete dictionaries; extending these results
to RIP matrices is left as a major open problem.

The main result in this paper is an algorithm that {\em provably} learns an unknown, incoherent dictionary from random samples $Y = AX$ where $X$ is a vector with at most $k \leq c \min(\sqrt{n}/\mu \log n, m^{1/2 - \eta})$ non-zero entries (for any $\eta > 0$, small enough constant $c>0$ depending on $\eta$). Hence we can allow almost as many non-zeros in the hidden vector $X$ as the best sparse recovery algorithms 
which {\em assume that the dictionary $A$ is known}. 
The precise requirements that we place on the distributional model are described in Section~\ref{sec:results}. We can relax some of these conditions at the cost of increased running time or requiring $X$ to be more sparse. Finally,  our algorithm can tolerate a substantial amount of additive noise, an important consideration in most  applications including sparse coding, provided it is independent and uncorrelated with the dictionary.

%

\medskip

\subsection{Related Works}
\paragraph{Algorithms used in practice}
Dictionary learning is solved in practice by variants of {\em alternating minimization}.  \cite{LS} gave the first approach;
subsequent popular approaches include the {\em method of optimal directions} (MOD) of
\cite{EAH}, and {\em K-SVD} of  \cite{AEB}. The general idea is to maintain a guess for $A$ and $X$  and at every 
step either update $X$ (using basis pursuit) or update $A$ by, say, solving a least squares problem. Provable guarantees for such algorithms have proved difficult because the initial guesses may be very far from the true dictionary, causing basis pursuit to behave erratically.  Also, the algorithms could converge to a dictionary that is not incoherent, and thus unusable for sparse recovery. (In practice, these heuristics do often work.)

\paragraph{Algorithms with guarantees} An elegant paper of \cite{SWW} shows how to provably recover $A$ {\em exactly} if it has full column rank, and $X$ has at most $\sqrt{n}$
nonzeros. 
However, requiring $A$ to be full column rank precludes most interesting applications where the dictionary is redundant and hence cannot have full column rank (see~\cite{DH,E, M}).  Moreover, the algorithm in~\cite{SWW} is not noise tolerant.

After the initial announcement of this work, \cite{AAN, AAJNT} independently gave provable algorithms for learning overcomplete and incoherent dictionaries. Their first paper~\citep{AAN} requires the entries in $X$ to be independent random $\pm 1$ variables. Their second~\citep{AAJNT} gives an algorithm --a version of alternating minimization-- that converges to the correct dictionary given a good initial dictionary (such a good initialization can only be found using \cite{AAN} in special cases, or more generally using this paper). Unlike our algorithms, theirs assume the sparsity of $X$ is at most $n^{1/4}$ or $n^{1/6}$ (assumption A4 in both papers), which are far from the $n^{1/2}$ limit of incoherent dictionaries. The main change from the initial version of our paper is that we have improved the dependence of our algorithms from $\mbox{poly}(1/\epsilon)$ to $\log 1/\epsilon$ (see Section~\ref{sec:refine}).

After this work, \cite{BKS} give an quasi-polynomial time algorithm for dictionary learning using sum-of-squares SDP hierarchy. The algorithm can output an approximate dictionary even when sparsity is almost linear in the dimensions with weaker assumptions. 

\paragraph{Independent Component Analysis}
When the entries of $X$ are independent, algorithms for {\em independent component analysis or ICA}~\citep{C} can recover $A$.
\cite{FJK} gave a provable algorithm that recovers $A$ up to arbitrary accuracy, provided entries in  $X$ are non-Gaussian (when $X$ is Gaussian, $A$ is only determined up to rotations anyway). Subsequent works considered the overcomplete case and gave provable algorithms even when $A$ is $n \times m$ with $m >n$~\citep{DCC,GVX}.

However, these algorithms are incomparable to ours since the algorithms are relying on different assumptions (independence vs. sparsity). With sparsity assumption, we can make much weaker assumptions on how $X$ is generated.  In particular, all these algorithms require the support $\Omega$ of the vector $X$ to be at least 3-wise independent ($\Pr[u,v,w\in \Omega] = \Pr[u\in \Omega]\Pr[v\in \Omega]\Pr[w\in \Omega]$) in the undercomplete case and $4$-wise independence in the overcomplete case . Our algorithm only requires the support $S$ to have bounded moments ($\Pr[u,v,w\in \Omega] \le \Lambda \Pr[u\in \Omega]\Pr[v\in \Omega]\Pr[w\in \Omega]$ where $\Lambda$ is a large constant or even a polynomial depending on $m,n,k$, see Definition~\ref{def:weak}).
Also, because our algorithm relies on the sparsity constraint, we are able to get almost exact recover in the noiseless case (see Theorem~\ref{thm:mainfront} and Section~\ref{sec:refine}). This kind of guarantee is impossible for ICA without sparsity assumption.

\subsection{Our Results}\label{sec:results}

A range of results are possible which trade off more assumptions  with  better performance.
We give two illustrative ones: the first makes the most assumptions but has the best performance;
the second has the weakest assumptions and somewhat worse performance. 
The theorem statements will be cleaner if we use asymptotic notation: the parameters $k, n, m$ will go to infinity and the constants denoted as ``$O(1)$" are arbitrary so long as they do not grow with these parameters. 

First we define the class of distributions that the $k$-sparse vectors must be drawn from.
We will be interested in distributions on $k$-sparse vectors in $\R^m$ where 
each coordinate is nonzero with probability $\Theta(k/m)$ (the constant in $\Theta(\cdot)$
can differ among coordinates). 

\begin{definition}[Distribution class $\Gamma$ and its moments] \label{def:gamma}
The distribution is in class $\Gamma$ if (i) each nonzero $X_i$ has expectation $0$ and lies in $[-C, -1] \cup [1, C]$ where $C =O(1)$. (ii) Conditioned on any subset of coordinates 
in $X$ being nonzero, the values $X_i$ are independent of each other.  

The distribution has {\em bounded $\ell$-wise moments} if the probability that $X$ is nonzero in any 
subset $S$  of $\ell$ coordinates is at most $c^{\ell}$ times 
$\prod_{i \in S}\Pr[X_i \neq 0]$ where $c =O(1)$.
\end{definition}

\paragraph{Remark:}(i) The bounded moments condition trivially holds for any constant $\ell$ if 
the set of nonzero locations is a random subset of size $k$. The {\em values} of these nonzero locations are allowed to be distributed very differently from one another. (ii) The requirement that nonzero $X_i$'s be bounded away from zero in magnitude is similar in spirit to the {\em Spike-and-Slab Sparse Coding} (S3C) model
of \cite{GCB}, which also encourages nonzero latent variables to be bounded away from zero to avoid degeneracy issues that arise when some coefficients are much larger than others.
(iii) In the rest of the paper we will be focusing on the case when $C = 1$, all the proofs generalize directly to the case $C > 1$ by losing constant factors in the guarantees.

Because of symmetry in the problem, we can only hope to learn dictionary $A$ up to permutation and sign-flips. We say two dictionaries are column-wise $\epsilon$-close, if after appropriate permutation and flipping the corresponding columns are within distance $\epsilon$.

\begin{definition}
Two dictionaries $A, B \in \R^{n\times m}$ are column-wise $\epsilon$-close, if there exists a permutation $\pi$ and $\theta \in \{\pm 1\}^m$ such that $\| (A_i) - \theta_i (B)_{\pi(i)}\| \le \epsilon$.
\end{definition}

Later when we are talking about two dictionaries that are $\epsilon$-close, we always assume the columns are ordered correctly so that $\|A_i - B_i\| \le \epsilon$.

\begin{theorem}
\label{thm:mainfront}
There is a polynomial time algorithm to learn a $\mu$-incoherent dictionary $A$ from random examples. With high probability the algorithm returns a dictionary $\hat{A}$ that is column-wise $\epsilon$ close to $A$ 
given random samples of the form $Y = AX$, where $X \in \R^n$ is chosen according to some distribution in $\Gamma$
and $A$ is in $\R^{n \times m}$:

\begin{itemize} \itemsep=0pt

\item If $k\leq c\min(m^{2/5}, \frac{\sqrt{n}}{\mu \log n})$ and the distribution has bounded $3$-wise moments, $c>0$ is a universal constant, then the algorithm requires $p_1$ samples 
and runs in time $\widetilde{O}(p_1^2n)$.

\item If $k \leq c\min(m^{(\ell -1)/(2\ell -1)}, \frac{\sqrt{n}}{\mu \log n})$ and the distribution has bounded $\ell$-wise moments, $c>0$ is a constant only depending on $\ell$, then the algorithm requires $p_2$ samples and runs in time  $\widetilde{O}(p_2^2 n)$


\item Even if each sample is of the form $Y^{(i)} = AX^{(i)} + \eta_i$, where $\eta_i$'s are independent spherical Gaussian noise with standard deviation $\sigma = o(\sqrt{n})$, the algorithms above still succeed provided the number of samples is at least $p_3$ and $p_4$ respectively.


\end{itemize}

\noindent In particular $p_1 = \Omega((m^2/k^2) \log m +  m k^2 \log m + m \log m \log 1/\epsilon)$ and $p_2 = \Omega((m/k)^{\ell -1} \log m + m k^2 \log m  \log 1/\epsilon)$ and $p_3 $ and $p_4$ are larger by a $\sigma^2/\epsilon^2$ factor. 

\end{theorem}


\paragraph{Remark:}
The sparsity that our algorithm can tolerate -- the minimum of $\frac{\sqrt{n}}{\mu \log n}$ and $m^{1/2 - \eta}$ -- approaches the sparsity that the best known algorithms require {\em even if $A$ is known}.  

\noindent Although the running time and sample complexity of the algorithm are relatively large polynomials, there are many ways to optimize the algorithm. See the discussion in Section~\ref{sec:implement}. 




Now we describe the other result which requires fewer assumptions on how the samples are generated, but require more stringent bounds on the sparsity:

\begin{definition}[Distribution class ${\mathcal D}$] \label{def:weak}
A distribution is in class ${\mathcal D}$
if (i) the events $X_i \ne 0$ have {\em weakly bounded} second and third moments, in the sense that $\Pr[X_i \ne 0\mbox{ and }X_j \ne 0] \le n^{\epsilon} \Pr[X_i\ne 0] \Pr[X_j\ne 0]$, $\Pr[X_i,X_j,X_t \ne 0] \le o(n^{1/4}) \Pr[X_i\ne 0]\Pr[X_j\ne 0]\Pr[X_t\ne 0]$.
(ii) Each nonzero $X_i$ is in $[-C, -1] \cup [1,C]$ where $C = O(1)$.
\end{definition}

The following theorem is proved similarly to Theorem~\ref{thm:mainfront}, and is sketched in Appendix~\ref{sec:ext}.

\begin{theorem}\label{thm:weakerassumption}
There is a polynomial time algorithm to learn a $\mu$-incoherent dictionary $A$ from random examples of the form $Y = AX$, where $X$ is chosen according to some distribution in
${\mathcal D}$. If $k\leq c\min(m^{1/4}, \frac{n^{1/4 - \epsilon/2}}{\sqrt{\mu}})$ and we are given $p \geq \Omega(\max(m^2/k^2 \log m,  \frac{m n^{3/2} \log m \log n }{k^2 \mu}))$ samples , then the algorithm succeeds with high probability, and the output dictionary is column-wise $\epsilon = O(k\sqrt{\mu}/n^{1/4-\epsilon/2})$ close to the true dictionary. The algorithm runs in time $\tilde{O}(p^2n + m^2p)$. 
The algorithm is also noise-tolerant as in Theorem~\ref{thm:mainfront}.
\end{theorem}

\subsection{Proof Outline}

%


The key observation in the algorithm is that we can test whether two samples share the same dictionary element (see Section~\ref{sec:graph}). Given this information, we can build a graph whose vertices are the samples, and edges correspond to samples that share the same dictionary element. A large cluster in this graph corresponds to the set of all samples with $X_i \ne 0$. In Section~\ref{sec:overlap} we give an algorithm for finding all the large clusters. Then we show how to recover the dictionary given the clusters in Section~\ref{sec:recovery}. This allows us to get a rough estimate of the dictionary matrix. Section~\ref{sec:refine} gives an algorithm for refining the solution in the noiseless case. The three main parts of the techniques are:

\medskip
\noindent{\bf Overlapping Clustering:} Heuristics such as MOD \citep{EAH} or K-SVD \citep{AEB} have a cyclic dependence: If we knew $A$, we could solve for $X$ and if we knew all of the $X$'s we could solve for $A$. Our main idea is to break this cycle by (without knowing $A$) finding all of the samples where $X_i \neq 0$. We can think of this as a cluster $\mathcal{C}_i$. Although our strategy is to cluster a random graph, what is crucial is that we are looking for an {\em overlapping} clustering since each sample $X$ belongs to $k$ clusters! Many of the algorithms which have been designed for finding overlapping clusterings (e.g.  \cite{AGSS}, \cite{BBBCT}) have a poor dependence on the maximum number of clusters that a node can belong to. Instead, we give a simple combinatorial algorithm based on triplet (or higher-order) tests that recovers the underlying, overlapping clustering. In order to prove correctness of our combinatorial algorithm, we rely on tools from discrete geometry, namely the {\em piercing number} \citep{Mat,AK}.

\medskip 
\noindent{\bf Recovering the Dictionary:} Next, we observe that there are a number of natural algorithms for recovering the dictionary once we know the clusters $\mathcal{C}_i$. We can think of a random sample from $\mathcal{C}_i$ as applying a filter to the samples we are given, and filtering out only those samples where $X_i \neq 0$. The claim is that this distribution will have a much larger variance along the direction $A_i$ than along other directions, and this allows us to recovery the dictionary either using a certain averaging algorithm, or by computing the largest singular vector of the samples in $\mathcal{C}_i$. In fact, this latter approach is similar to K-SVD \citep{AEB} and hence our analysis yields insights into why these heuristics work. 

\medskip 
\noindent{\bf Fast Convergence:} The above approach yields provable algorithms for dictionary learning whose running time and sample complexity depend polynomially on $1/\epsilon$. However once we have a suitably good approximation to the true dictionary, can we converge at a much faster rate? We analyze a simple alternating minimization algorithm {\sc Iterative Average} and we derive a formula for its updates where we can analyze it by thinking of it instead as a noisy version of the matrix power method (see Lemma~\ref{lemma:noisypower}). This analysis is inspired by recent work on analyzing alternating minimization for the matrix completion problem \citep{JNS,H}, and we obtain algorithms whose running time and sample complexity depends on $\log 1/\epsilon$. Hence we get algorithms that converge rapidly to the true dictionary while simultaneously being able to handle almost the same sparsity as in the sparse recovery problem where $A$ is known!

\medskip

\noindent {\bf NOTATION:}
Throughout this paper, we will use $Y^{(i)}$ to denote the $i^{th}$ sample and $X^{(i)}$ as the vector that generated it -- i.e. $Y^{(i)} = A X^{(i)}$. Let $\Omega^{(i)}$ denote the support of $X^{(i)}$. For a vector $X$ let $X_i$ be the $i^{th}$ coordinate.
For a matrix $A \in \R^{n\times m}$ (especially the dictionary matrix), we use $A_i$ to denote the $i$-th column (the $i$-th dictionary element). Also, for a set $S\subset \{1,2,...,m\}$, we use $A_S$ to denote the submatrix of $A$ with columns in $S$.
We will use $\|A\|_F$ to denote the Frobenius norm and $\|A\|$ to denote the spectral norm. Moreover we will use $\Gamma$ to denote the distribution on $k$-sparse vectors $X$ that is used to generate our samples, and $\Gamma_i$ will denote the restriction of this distribution to vectors $X$ where $X_i \neq 0$. When we are working with a graph $G$ we will use $\Gamma_G(u)$ to denote the set of neighbors of $u$ in $G$. 
Throughout the paper ``with high probability'' means the probability is at least $1-n^{-\Delta}$ for large enough $\Delta$.

\section{The Connection Graph}\label{sec:graph}

In this part we show how to test whether two samples share the same dictionary element, i.e., whether the supports $\Omega^{(i)}$ and $\Omega^{(j)}$ intersect. The idea is we can check the inner-product of $Y^{(i)}$ and $Y^{(j)}$, which can be decomposed into the sum of inner-products of dictionary elements $$\langle Y^{(i)}, Y^{(j)}\rangle = \sum_{p\in \Omega^{(i)},q\in\Omega^{(j)}} \langle A_p,A_q\rangle X^{(i)}_pX^{(j)}_q$$ If the supports are disjoint, then each of the terms above is small since $\langle A_p,A_q\rangle \leq \mu / \sqrt{n}$ by the incoherence assumption. 
To prove the sum is indeed small, we will appeal to the classic Hanson-Wright inequality:
\begin{theorem} [Hanson-Wright] \citep{HW}
Let $X$ be a vector of independent, sub-Gaussian random variables with mean zero and variance one. Let $M$ be a symmetric matrix. Then $$ Pr[|X^T M X - tr(M)| > t] \leq 2 exp\{-c \min(t^2/\|M\|_F^2, t/\|M\|_2)\}$$
\end{theorem}

This will allow us to determine if $\Omega^{(1)}$ and $\Omega^{(2)}$ intersect but with false negatives:

\begin{lemma}\label{lemma:hw}
Suppose $k \mu < \frac{\sqrt{n}}{C' \log n}$ for large enough constant $C'$ (depending on $C$ in Definition~\ref{def:gamma}). Then if $\Omega^{(i)}$ and $\Omega^{(j)}$ are disjoint, with high probability $| \langle Y^{(i)}, Y^{(j)} \rangle | < 1/2$.
\end{lemma}

\begin{proof}
Let $N$ be the $k \times k$ submatrix resulting from restricting $A^T A$ to the locations where $X^{(i)}$ and $X^{(j)}$ are non-zero. Set $M$ to be a $2k \times 2k$ matrix where the $k \times k$ submatrices in the top-left and bottom-right are zero, and the $k \times k$ submatrices in the bottom-left and top-right are $(1/2) N$ and $(1/2) N^T$ respectively. Here we think of the vector $X$ as being a length $2k$ vector whose first $k$ entries are the non-zero entries in $X^{(i)}$ and whose last $k$ entries are the non-zero entries in $X^{(j)}$. And by construction, we have that $$\langle Y^{(i)}, Y^{(j)} \rangle = X^T M X$$

We can now appeal to the Hanson-Wright inequality (above). Note that since $\Omega^{(i)}$ and $\Omega^{(j)}$ do not intersect, the entries in $M$ are each at most $\mu/\sqrt{n}$ and so the Frobenius norm of $M$ is at most $\frac{\mu k}{ \sqrt{2n}}$. This is also an upper-bound on the spectral norm of $M$. We can set $t = 1/2$, and for $k \mu < \sqrt{n}/C' \log n$ both terms in the minimum are $\Omega(\log n)$ and this implies the lemma. 
\end{proof}

We will also make use of a weaker bound (but whose conditions allow us to make fewer distributional assumptions):  

\begin{lemma}\label{lemma:graph}
If $k^2 \mu < \sqrt{n}/2$ then $| \langle Y^{(i)}, Y^{(j)} \rangle | > 1/2$ implies that $\Omega^{(i)}$ and $\Omega^{(j)}$ intersect
\end{lemma}

\begin{proof}
Suppose $\Omega^{(i)}$ and $\Omega^{(j)}$ are disjoint. Then the following upper bound holds: $$ | \langle Y^{(i)}, Y^{(j)} \rangle | \leq \sum_{p \neq q} | \langle A_p, A_q \rangle  X^{(i)}_p X^{(j)}_q | \leq k^2 \mu / \sqrt{n} < 1/2$$ and this implies the lemma.
\end{proof}

\noindent This only works up to $k = O(n^{1/4} / \sqrt{\mu})$. In comparison, the stronger bound of Lemma~\ref{lemma:hw} makes use of the randomness of the signs of $X$ and works up to $k = O(\sqrt{n}/\mu \log n)$.

In our algorithm, we build the following graph:

\begin{definition}
Given $p$ samples $Y^{(1)}, Y^{(2)}, ..., Y^{(p)}$, build a {\em connection graph} on $p$ nodes where $i$ and $j$ are connected by an edge if and only if $| \langle Y^{(i)}, Y^{(j)} \rangle | > 1/2$.
\end{definition}

\noindent This graph will ``miss" some edges, since if a pair $X^{(i)}$ and $X^{(j)}$ have intersecting support we do not necessarily meet the above condition. But by Lemma~\ref{lemma:hw} (with high probability) this graph will not have any false positives:

\begin{corollary}
With high probability, each edge $(i, j)$ present in the connection graph corresponds to a pair where $\Omega^{(i)}$ and $\Omega^{(j)}$ have non-empty intersection.
\end{corollary}

Consider a sample $Y^{(1)}$ for which there is an edge to both $Y^{(2)}$ and $Y^{(3)}$. This means that there is some coordinate $i$ in both $\Omega^{(1)}$ and $\Omega^{(2)}$ and some coordinate $i'$ in both $\Omega^{(1)}$ and $\Omega^{(3)}$. However the challenge is that we do not immediately know if  $\Omega^{(1)}, \Omega^{(2)}$ and $\Omega^{(3)}$ have a common intersection or not. 


\section{Overlapping Clustering}\label{sec:overlap}

Our goal in this section is to determine which samples $Y$ have $X_i \neq 0$ just from the connection graph. To do this, we will identify a combinatorial condition that allows us to decide whether or not a set of three samples $Y^{(1)}, Y^{(2)}$ and $Y^{(3)}$ that have supports $\Omega^{(1)}, \Omega^{(2)}$ and $\Omega^{(3)}$ respectively -- have a common intersection or not. From this condition, it is straightforward to give an algorithm that correctly groups together all of the samples $Y$ that have $X_i \neq 0$. In order to reduce the number of letters used we will focus on the first three samples $Y^{(1)}, Y^{(2)}$ and $Y^{(3)}$ although all the claims and lemmas hold for all triples.

Suppose we are given two samples $Y^{(1)}$ and $Y^{(2)}$ with supports $\Omega^{(1)}$ and $\Omega^{(2)}$ where $\Omega^{(1)} \cap \Omega^{(2)} = \{i\}$. We will prove that this pair can be used to recover all the samples $Y$ for which $X_i \neq 0$. This will follow because we will show that  the expected number of common neighbors between $Y^{(1)}, Y^{(2)}$ and $Y$ will be large if $X_i \neq 0$ and otherwise will be small. So throughout this subsection let us consider a sample $Y = AX$ and let $\Omega$ be its support. We will need the following elementary claim. 

\begin{claim}\label{claim:upper}
Suppose $\Omega^{(1)} \cap \Omega^{(2)} \cap \Omega^{(3)} \neq \emptyset$, then $Pr_Y[ \mbox{ for all } j = 1, 2, 3, |\langle Y, Y^{(j)} \rangle | > 1/2] \geq \frac{k}{2m}$
\end{claim}

\begin{proof}
Using ideas similar to Lemma~\ref{lemma:hw}, we can show if $|\Omega\cap \Omega^{(1)}| = 1$ (that is, the new sample has a {\em unique} intersection with $\Omega^{(1)}$), then $|\langle Y, Y^{(1)} \rangle| > 1/2$.

Now let $i \in \Omega^{(1)} \cap \Omega^{(2)} \cap \Omega^{(3)}$, let $\mathcal{E}$ be the event that $\Omega\cap \Omega^{(1)} = \Omega \cap \Omega^{(2)}=\Omega \cap \Omega^{(3)} = \{i\}$. Clearly, when event $\mathcal{E}$ happens, for all $j = 1, 2, 3, |\langle Y, Y^{(j)} \rangle | > 1/2$. The probability of $\mathcal{E}$ is at least $$\Pr[i\in \Omega] \Pr[(\Omega^{(1)} \cup \Omega^{(2)} \cup \Omega^{(3)}\backslash\{i\})\cap \Omega = \emptyset | i\in \Omega] = k/m \cdot (1-O(k/m)\cdot 3k) \ge k/2m.$$

Here we used bounded second moment property for the conditional probability and union bound.
\end{proof}

This claim establishes a lower bound on the expected number of common neighbors of a triple, if they have a common intersection. Next we establish an upper bound, if they don't have a common intersection. Suppose $\Omega^{(1)} \cap \Omega^{(2)} \cap \Omega^{(3)} = \emptyset$. In principle we should be concerned that $\Omega$ could still intersect each of $\Omega^{(1)}$, $\Omega^{(2)}$ and $ \Omega^{(3)}$ in different locations.
Let $a = |\Omega^{(1)} \cap \Omega^{(2)}|$, $b = |\Omega^{(1)} \cap \Omega^{(3)} |$ and $c = |\Omega^{(2)} \cap \Omega^{(3)} |$.

\begin{lemma}\label{lemma:lower}
Suppose that $\Omega^{(1)} \cap \Omega^{(2)} \cap \Omega^{(3)} = \emptyset$. Then the probability that $\Omega$ intersects each of $\Omega^{(1)}$, $\Omega^{(2)}$ and $ \Omega^{(3)}$ is at most $$ \frac{k^6}{m^3} + \frac{3k^3(a+b+c)}{m^2}$$
\end{lemma}

\begin{proof}
We can break up the event whose probability we would like to bound into two (not necessarily disjoint) events: (1) the probability that $\Omega$ intersects each of $\Omega^{(1)}$, $\Omega^{(2)}$ and $ \Omega^{(3)}$ disjointly (i.e. it contains a point $i \in \Omega^{(1)}$ but $i \notin \Omega^{(2)}, \Omega^{(3)}$, and similarly for the other sets ). (2) the probability that $\Omega$ contains a point in the common intersection of two of the sets, and one point from the remaining set. Clearly if $\Omega$ intersects the each of $\Omega^{(1)}$, $\Omega^{(2)}$ and $ \Omega^{(3)}$ then at least one of these two events must occur.

The probability of the first event is at most the probability that $\Omega$ contains at least one element from each of three disjoint sets of size at most $k$. The probability that $\Omega$ contains an element of just one such set is at most the expected intersection which is $\frac{k^2}{m}$, and since the expected intersection of $\Omega$ with each of these sets are non-positively correlated (because they are disjoint) we have that the probability of the first event can be bounded by $\frac{k^6}{m^3}$. 

Similarly, for the second event: consider the probability that $\Omega$ contains an element in $\Omega^{(1)} \cap \Omega^{(2)}$. Since $\Omega^{(1)} \cap \Omega^{(2)} \cap \Omega^{(3)} = \emptyset$, then $\Omega$ must also contain an element in $\Omega^{(3)}$ too. The expected intersection of $\Omega$ and $\Omega^{(1)} \cap \Omega^{(2)}$ is $\frac{ka}{m}$ and the expected intersection of $\Omega$ and $\Omega^{(3)}$ is $\frac{k^2}{m}$, and again the expectations are non-positively correlated since the two sets $\Omega^{(1)} \cap \Omega^{(2)}$ and $\Omega^{(3)}$ are disjoint by assumption. Repeating this argument for the other pairs completes the proof of the lemma. 
\end{proof}

\noindent Note that if $\Gamma$ has bounded higher order moment, the probability that two sets of size $k$ intersect in at least $Q$ elements is at most $(\frac{k^2}{m})^Q$. Hence we can assume that with high probability there is {\em no} pair of samples whose supports intersect in more than a constant number of locations. When $\Gamma$ only has bounded 3-wise moment see Appendix~\ref{asec:boundedmoment}.


Let us quantitatively compare our lower and upper bound: If $k \leq c m^{2/5}$ then the expected number of common neighbors for a triple with $\Omega^{(1)} \cap \Omega^{(2)} \cap \Omega^{(3)} \neq \emptyset$ is much larger than the expected number of common neighbors of a triple whose common intersection is empty. Under this condition, if we take $p = O(m^2/k^2 \log n)$ samples each triple with a common intersection will have at least $T$ common neighbors, and each triple whose common intersection is empty will have less than $T/2$ common neighbors. 

Hence we can search for a triple with a common intersection as follows: We can find a pair of samples $Y^{(1)}$ and $Y^{(2)}$ whose supports intersect. We can take a neighbor $Y^{(3)}$ of $Y^{(1)}$ in the connection graph (at random), and by counting the number of common neighbors of $Y^{(1)}$, $Y^{(2)}$ and $Y^{(3)}$ we can decide whether or not their supports have a common intersection. 

\begin{fragment*}[t]
\caption{
\label{alg:ocluster}{\sc OverlappingCluster}, \textbf{Input: } $p$ samples $Y^{(1)}, Y^{(2)}, ..., Y^{(p)}$ 
}
\begin{enumerate} \itemsep 0pt
\small 
\item Compute a graph $G$ on $p$ nodes where there is an edge between $i$ and $j$ iff $|\langle Y^{(i)}, Y^{(j)} \rangle | > 1/2$
\item Set $T = \frac{pk}{10m}$
\item Repeat $\Omega(m \log^2 m)$ times:
\item $\qquad$ Choose a random edge $(u, v)$ in $G$
\item $\qquad$ Set $S_{u,v} = \{w : |\Gamma_G(u) \cap \Gamma_G(v) \cap \Gamma_G(w)| \geq T \} \cup \{u, v\}$
\item Delete any set $S_{u,v}$ where $u, v$ are contained in a strictly smaller set $S_{a,b}$ (also delete any duplicates)
\item Output the remaining sets $S_{u,v}$
\end{enumerate} 
\end{fragment*}

\begin{definition}
\label{def:identifyingpair}
We will call a pair of samples $Y^{(1)}$ and $Y^{(2)}$ an {\em identifying pair} for coordinate $i$ if the intersection of $\Omega^{(1)}$ and $\Omega^{(2)}$  is exactly $\{i\}$. 
\end{definition}

\begin{theorem}\label{thm:ocluster}
The output of {\sc OverlappingCluster} is an overlapping clustering where each set corresponds to some $i$ and contains all $Y^{(j)}$ for which $i \in \Omega^{(j)}$. The algorithm runs in time $\widetilde{O}(p^2 n)$ and succeeds with high probability if $k \leq c \min(m^{2/5}, \frac{\sqrt{n}}{\mu \log n})$ and if $p =  \Omega(\frac{m^2 \log m}{k^2})$ 
\end{theorem}

\begin{proof}
We can use Lemma~\ref{lemma:hw} to conclude that each edge in $G$ corresponds to a pair whose support intersects. We can appeal to Lemma~\ref{lemma:lower} and Claim~\ref{claim:upper} to conclude that for $p =  \Omega(m^2/k^2 \log m)$, with high probability each triple with a common intersection has at least $T$ common neighbors, and each triple without a common intersection has at most $T/2$ common neighbors. 

In fact, for a random edge $(Y^{(1)}, Y^{(2)})$, the probability that the common intersection of $\Omega^{(1)}$ and $\Omega^{(2)}$ is exactly $\{i\}$ is $\Omega(1/m)$ because we know that they do intersect, and that intersection has a constant probability of being size one and it is uniformly distributed over $m$ possible locations. 
Appealing to a coupon  collector argument we conclude that if the inner loop is run at least $\Omega(m \log^2 m)$ times then the algorithm finds an identifying pair $(u, v)$ for each column $A_i$ with high probability. 

Note that we may have pairs that are not an identifying triple for some coordinate $i$. However, any other pair $(u, v)$ found by the algorithm must have a common intersection. Consider for example a pair $(u, v)$ where $u$ and $v$ have a common intersection $\{i, j\}$. Then we know that there is some other pair $(a, b)$ which is an identifying pair for $i$ and hence $S_{a, b} \subset S_{u, v}$. (In fact this containment is strict, since $S_{u, v}$ will also contain a set corresponding to an identifying pair for $j$ too). Hence the second-to-last step in the algorithm will necessarily delete all such non-identifying pairs $S_{u, v}$. 

What is the running time of this algorithm? We need $O(p^2 n)$ time to build the connection graph, and the loop takes $\widetilde{O}(pmn)$ time. Finally, the deletion step requires time $\widetilde{O}(m^2)$ since there will be $\widetilde{O}(m)$ pairs found in the previous step and for each pair of pairs, we can delete $S_{u, v}$ if and only if there is a strictly smaller $S_{a, b}$ that contains $u$ and $v$. This concludes the proof of correctness of the algorithm, and its running time analysis. 
\end{proof}

\section{Recovering the Dictionary}
\label{sec:recovery}
\subsection{Finding the Relative Signs}\label{sec:recovery1}


Here we show how to recover the column $A_i$ once we have learned which samples $Y$ have $X_i \neq 0$. We will refer to this set of samples as the ``cluster" $\mathcal{C}_i$. The key observation is that if $\Omega^{(1)}$ and $\Omega^{(2)}$ uniquely intersect in index $i$ then the sign of $\langle Y^{(1)}, Y^{(2)}\rangle$ is equal to the sign of $X^{(1)}_i X^{(2)}_i$. If there are enough such pairs $Y^{(1)}$ and $Y^{(2)}$, we can determine not only which samples $Y$ have $X_i \neq 0$ but also which pairs of samples $Y$ and $Y'$ have $X_i, X'_i \neq 0$ and $\mbox{sign}(X_i) = \mbox{sign}(X'_i)$. This is the main step of the algorithm {\sc OverlappingAverage}.

\begin{theorem}\label{thm:avgmain}
If the input to {\sc OverlappingAverage} $\mathcal{C}_1,...,\mathcal{C}_m$ are the true clusters $\{j:i\in \Omega^{(j)}\}$ up to permutation, 
then the algorithm outputs a dictionary $\hat{A}$ that is column-wise $\epsilon$-close to $A$ 
with high probability if $k \leq  \min(\sqrt{m}, \frac{\sqrt{n}}{\mu })$ and if $p = \Omega \left ( \max(m^2\log m /k^2 ,  m \log m / \epsilon^2) \right )$
Furthermore the algorithm runs in time $O(p^2)$.
\end{theorem}

Intuitively, the algorithm works because the sets $\mathcal{C}^\pm_i$ correctly identify samples with the same sign. This is summarized in the following lemma.

\begin{lemma}\label{lemma:signset}
In Algorithm~\ref{alg:oaverage}, $\mathcal{C}^{\pm}_i$ is either $\{u:X^{(u)}_i > 0\}$ or $\{u: X^{(u)}_i < 0\}$.
\end{lemma}

\begin{proof}
It suffices to prove the lemma at the start of Step~\ref{item:reverse}, since this step only takes the complement of $\mathcal{C}^{\pm}_i$ with respect to $\mathcal{C}_i$. Appealing to Lemma~\ref{lemma:hw} we conclude that $\Omega^{(u)}$ and $\Omega^{(v)}$ uniquely intersect in coordinate $i$ then the sign of $\langle Y^{(u)}, Y^{(v)}\rangle$  is equal to the sign of $X^{(u)}_i X^{(v)}_i$. Hence when Algorithm~\ref{alg:oaverage} adds an element to $\mathcal{C}^{\pm}_i$ it must have the same sign as the $i^{th}$ component of $X^{(u_i)}$. What remains is to prove that each node $v \in \mathcal{C}_i$ is correctly labeled. We will do this by showing that for any such vertex, there is a length two path of labeled pairs that connects $u_i$ to $v$, and this is true because the number of labeled pairs is large. We need the following simple claim:

\begin{claim}
\label{claim:intersect}
If $p > m^2\log m/k^2$ then with high probability any two clusters share at most $2pk^2/m^2$ nodes in common.
\end{claim}

This follows since the probability that a node is contained in any fixed pair of clusters is at most $k^2/m^2$. Then for any node $u \in \mathcal{C}_i$, we would like to lower bound the number of labeled pairs it has in $\mathcal{C}_i$. Since $u$ is in at most $k-1$ other clusters $\mathcal{C}_{i_1}, ..., \mathcal{C}_{i_{k-1}}$, the number of pairs $u,v$ where $v\in \mathcal{C}_i$ that are not labeled for $\mathcal{C}_i$ is at most $$\sum_{t=1}^{k-1} |\mathcal{C}_{i_t} \cap \mathcal{C}_i| \le k \cdot 2pk^2/m^2 \ll pk/3m = |\mathcal{C}_i|/3$$ Therefore for a fixed node $u$ for at least a $2/3$ fraction of the other nodes $w \in \mathcal{C}_i$ the pair $u, w$ is labeled. Hence we conclude that for each pair of nodes $u_i, v \in \mathcal{C}_i$ the number of $w$ for which both $u_i, w$ and $w, v$ are labeled is at least $|\mathcal{C}_i|/3 > 0$ and so for every $v$, there is a labeled path of length two connecting $u_i$ to $v$.
\end{proof}

\begin{fragment*}[t]
\caption{
\label{alg:oaverage}{\sc OverlappingAverage}, \textbf{Input: } $p$ samples $Y^{(1)}, Y^{(2)}, ... Y^{(p)}$ and overlapping clusters $\mathcal{C}_1,  \mathcal{C}_2, ..., \mathcal{C}_m$ \vspace*{0.01in}
}

\begin{enumerate} \itemsep 0pt
\small 
\item For each $\mathcal{C}_i$
\item $\qquad$ For each pair $(u,v)\in \mathcal{C}_i$ that does not appear in any other $\mathcal{C}_j$ 
($X^{(u)}$ and $ X^{(v)}$ have a unique intersection)
\item $\qquad$ $\qquad$ Label the pair $+1$ if $\langle Y^{(u)},Y^{(v)}\rangle > 0$ and otherwise label it $-1$.
\item $\qquad$ Choose an arbitrary $u_i\in \mathcal{C}_i$, and set $\mathcal{C}^{\pm}_i = \{u_i\}$
\item $\qquad$ For each $v \in \mathcal{C}_i$
\item $\qquad$ $\qquad$ If the pair $u_i, v$ is labeled $+1$ add $v$ to $\mathcal{C}^{\pm}_i$
\item $\qquad$ $\qquad$ Else if there is $w \in \mathcal{C}_i$ where the pairs $u_i, w$ and $v, w$ have the same label, add $v$ to $\mathcal{C}^{\pm}_i$.
\item {\label{item:reverse}}$\qquad$ If $|\mathcal{C}^{\pm}_i| \le |\mathcal{C}_i|/2$ set $\mathcal{C}^{\pm}_i \leftarrow \mathcal{C}_i\backslash \mathcal{C}^{\pm}_i$.
\item $\qquad$ Let $\hat{A}_i = \sum_{v\in \mathcal{C}^{\pm}_i} Y^{(v)} / \|\sum_{v\in \mathcal{C}^{\pm}_i} Y^{(v)}\|$
\item Output $\hat{A}$, where each column is $\hat{A}_i$ for some $i$
\end{enumerate} 

\end{fragment*}

Using this lemma, we are ready to prove Algorithm~\ref{alg:oaverage} correctly learns all columns of $A$. 

\begin{proof}
We can invoke Lemma~\ref{lemma:signset} and conclude that $\mathcal{C}^{\pm}_i$ is either $\{u:X^{(u)}_i > 0\}$ or $\{u:X^{(u)}_i < 0\}$, whichever set is larger. Let us suppose that it is the former. Then each  $Y^{(u)}$ in $\mathcal{C}^{\pm}_i$ is an independent sample from the distribution conditioned on $X_i > 0$, which we call $\Gamma_i^+$. We have that $\E_{\Gamma_i^+}[AX] = c A_i$ where $c$ is a constant in $[1,C]$ because $\E_{\Gamma_i^+}[X_j] = 0$ for all $j\ne i$.

Let us compute the variance:
$$\E_{\Gamma_i^+}[\|AX - \E_{\Gamma_i^+}AX\|^2] \le \E_{\Gamma_i^+} X_i^2 + \sum_{j\ne i} \E_{\Gamma_i^+} [X_j^2]
\le C^2 + \sum_{j\ne i} C^2 k/m \le C^2 (k+1),$$
\noindent Note that there are no cross-terms because the signs of each $X_j$ are independent. Furthermore we can bound the norm of each vector $Y^{(u)}$ via incoherence. We can conclude that if $|\mathcal{C}^{\pm}_i| > C^2k \log m/\epsilon^2$, then with high probability $\|\hat{A}_i - A_i\| \le \epsilon$ using vector Bernstein's inequality (\cite{G}, Theorem 12). This latter condition holds because we set $\mathcal{C}^{\pm}_i$ to itself or its complement based on which one is larger. 
\end{proof}

\subsection{An Approach via SVD}\label{sec:recovery2}

Here we give an alternative algorithm for recovering the dictionary based instead on SVD.
Intuitively if we take all the samples whose support contains index $j$, then every such sample $Y^{(i)}$ has a component along direction $A_j$. Therefore direction $A_j$ should have the largest variance and can be found by SVD. The advantage is that methods like K-SVD which are quite popular in practice also rely on finding directions of maximum variance, so the analysis we provide here yields insights into why these approaches work. However, the crucial difference is that we rely on finding the correct overlapping clustering in the first step of our dictionary learning algorithms, whereas K-SVD and approaches like approximate it via their current guess for the dictionary. 

 Let us fix some notation: Let $\Gamma_i$ be the distribution conditioned on $X_i \neq 0$. Then once we have found the overlapping clustering, each cluster is a set of random samples from $\Gamma_i$. Also let $\alpha = | \langle u, A_i \rangle|$. 

\begin{definition}\label{def:Ri}
Let $R_i^2 = 1 + \sum_{j \neq i} \langle A_i, A_j \rangle^2 E_{\Gamma_i}[X_j^2]$.
\end{definition}

\noindent Note that $R_i^2$ is the projected variance of $\Gamma_i$ on the direction $u = A_i$. Our goal is to show that for any $u \neq A_i$ (i.e. $\alpha \neq 1$), the variance is strictly smaller. 

\begin{lemma}\label{lemma:var}
The projected variance of $\Gamma_i$ on $u$ is at most $$\alpha^2R_i^2 + \alpha\sqrt{(1-\alpha^2)}\frac{2 \mu k}{\sqrt{n}} + (1-\alpha^2)(\frac{k}{m} + \frac{\mu k}{\sqrt{n}})$$
\end{lemma}

\begin{proof}
Let $u^{||}$ and $u^{\perp}$ be the components of $u$ in the direction of $A_i$ and perpendicular to $A_i$. Then we want bound $E_{\Gamma_i}[ \langle u, Y \rangle^2]$ where $Y$ is sampled from $\Gamma_i$. Since the signs of each $X_j$ are independent, we can write $$E_{\Gamma_i}[ \langle u, Y \rangle^2] = \sum_j E_{\Gamma_i}[ \langle u, A_j X_j \rangle^2] =  \sum_j E_{\Gamma_i}[ \langle u^{||} + u^{\perp}, A_j X_j \rangle^2] $$ Since $\alpha = \| u^{||}\|$ we have: $$E_{\Gamma_i}[ \langle u, Y \rangle^2] = \alpha^2 R_i^2 + E_{\Gamma_i}[\sum_{j \neq i} (2 \langle u^{||} , A_j \rangle \langle u^{\perp} , A_j \rangle + \langle u^{\perp} , A_j \rangle^2 ) X_j^2]$$ Also $E_{\Gamma_i}[X_j^2] = (k-1)/(m-1)$. 
Let $v$ be the unit vector in the direction $u^{\perp}$. We can write $$E_{\Gamma_i}[\sum_{j \neq i} \langle u^{\perp} , A_j \rangle^2  X_j^2] = (1-\alpha^2) ( \frac{k-1}{m-1} ) v^T A_{-i} A_{-i}^T v$$where $A_{-i}$ denotes the dictionary $A$ with the $i^{th}$ column removed. The maximum over $v$ of $v^T A_{-i} A_{-i}^T v$ is just the largest singular value of $A_{-i} A_{-i}^T$ which is the same as the largest singular value of $A_{-i}^T A_{-i}$ which by the Greshgorin Disk Theorem (see e.g. \cite{HJ}) is at most $1 + \frac{\mu}{\sqrt{n}} m$. And hence we can bound $$ E_{\Gamma_i}[\sum_{j \neq i} \langle u^{\perp} , A_j \rangle^2  X_j^2] \leq  (1- \alpha^2) (\frac{k}{m} + \frac{\mu k}{\sqrt{n}})$$ Also since $  |\langle u^{||} , A_j \rangle| = \alpha |\langle A_i , A_j \rangle| \leq \alpha \mu/\sqrt{n}$ we obtain:
 $$E[\sum_{j \neq i} 2 \langle u^{||} , A_j \rangle \langle u^{\perp} , A_j \rangle X_j^2] \leq \alpha \sqrt{(1-\alpha^2)}\frac{2 \mu k}{ \sqrt{n}}$$
and this concludes the proof of the lemma.
\end{proof}

\begin{definition}
Let $\zeta = \max\{\frac{\mu k}{\sqrt{n}}, \sqrt{\frac{k}{m}}\}$, so the expression in Lemma~\ref{lemma:var} can be  be an upper bounded by $\alpha^2R_i^2 + 2\alpha\sqrt{1-\alpha^2} \cdot \zeta  + (1-\alpha^2)\zeta^2$. 
\end{definition}

\begin{fragment*}[t]
\caption{
\label{alg:osvd}{\sc OverlappingSVD}, \textbf{Input: } $p$ samples $Y^{(1)}, Y^{(2)}, ... Y^{(p)}$ \vspace*{0.01in}
}

\begin{enumerate} \itemsep 0pt
\small 
\item Run {\sc OverlappingCluster} (or {\sc OverlappingCluster2}) on the $p$ samples
\item Let $\mathcal{C}_1,  \mathcal{C}_2, ... \mathcal{C}_m$ be the $m$ returned overlapping clusters
\item Compute $\hat{\Sigma}_i = \frac{1}{|\mathcal{C}_i|} \sum_{Y \in \mathcal{C}_i} YY^T$
\item Compute the first singular value $\hat{A}_i$ of $\hat{\Sigma}_i$
\item Output $\hat{A}$, where each column is $\hat{A}_i$ for some $i$
\end{enumerate} 

\end{fragment*}

We will show that an approach based on SVD recovers the true dictionary up to additive accuracy $\pm \zeta$. Note that here $\zeta$ is a parameter that converges to zero as the size of the problem increases, but is not a function of the number of samples. So unlike the algorithm in the previous subsection, we cannot make the error in our algorithm arbitrarily small by increasing the number of samples, but this algorithm has the advantage that it succeeds even when $\E[X_i] \neq 0$. 

\begin{corollary}\label{cor:sep}
The maximum singular value of $\Gamma_i$ is at least $R_i$ and the direction $u$ satisfies $\| u - A_i\| \leq O(\zeta)$. Furthermore the second largest singular value is bounded by $O(R_i^2 \zeta^2)$. 
\end{corollary}

\begin{proof}
The bound in Lemma~\ref{lemma:var} is only an upper bound, however the direction $\alpha =1$ has variance $R_i^2 > 1$ and hence the direction of maximum variance must correspond to $\alpha \in [1-O(\zeta^2), 1]$. Then we can appeal to the variational characterization of singular values (see \cite{HJ}) that $$\sigma_2(\Sigma_i) = \max_{u \perp A_i} \frac{u^T \Sigma_i u}{u^T u}$$ Then condition that $\alpha \in [-O(\zeta), O(\zeta)]$ for the second singular value implies the second part of the corollary. 
\end{proof}

Since we have a lower bound on the separation between the first and second singular values of $\Sigma_i$, we can apply Wedin's Theorem and show that we can recover $A_i$ approximately even in the presence of noise. 

\begin{theorem}[Wedin]\label{thm:sine} \citep{W}
Let $\delta = \sigma_1(M) - \sigma_2(M)$ and let $M' = M + E$ and furthermore let $v_1$ and $v_1'$ be the first singular vectors of $M$ and $M'$ respectively. Then $$\sin \Theta(v_1, v_1') \leq C \frac{\| E\|}{\delta}$$
\end{theorem}

\noindent Hence even if we do not have access to $\Sigma_i$ but rather an approximation to it $\hat{\Sigma}_i$ (e.g. an empirical covariance matrix computed from our samples), we can use the above perturbation bound to show that we can still recover a direction that is close to $A_i$ -- and in fact converges to $A_i$ as we take more and more samples. 

\begin{theorem}\label{thm:main}
If the input to {\sc OverlappingSVD} is the correct clustering, then the algorithm outputs a dictionary $\hat{A}$ so that for each $i$, $\|A_i - \hat{A}_i\| \leq \zeta$ with high probability if $k \leq c \min(\sqrt{m}, \frac{\sqrt{n}}{\mu \log n})$ and if $$p \geq \max(m^2 \log m /k^2,  \frac{m n \log m \log n}{\zeta^2})$$
\end{theorem}

\begin{proof}
Appealing to Theorem~\ref{thm:ocluster}, we have that with high probability the call to {\sc OverlappingCluster} returns the correct overlapping clustering. Then given $\frac{n \log n}{\zeta^2}$ samples from the distribution $\Gamma_i$ the classic result of Rudelson implies that the computed empirical covariance matrix $\hat{\Sigma}_i $ is close in spectral norm to the true co-variance matrix \cite{R}. This, combined with the separation of the first and second singular values established in Corollary~\ref{cor:sep} and Wedin's Theorem \ref{thm:sine} imply that we recover each column of $A$ up to an additive accuracy of $\epsilon$ and this implies the theorem. Note that since we only need to compute the first singular vector, this can be done via power iteration \cite{GV} and hence the bottleneck in the running time is the call to {\sc OverlappingCluster}. 
\end{proof}


\subsection{Noise Tolerance}

Here we elaborate on why the algorithm can tolerate noise provided that the noise is {\em uncorrelated} with the dictionary (e.g. Gaussian noise). The observation is that in constructing the connection graph, we only make use of the inner products between pairs of samples $Y^{(1)}$ and $Y^{(2)}$, the value of which is roughly preserved under various noise models. In turn, the overlapping clustering is a purely combinatorial algorithm that only makes use of the connection graph. Finally, we recover the dictionary $A$ using singular value decomposition, which is well-known to be stable under noise (e.g. Wedin's Theorem~\ref{thm:sine}). 

\section{Refining the Solution}\label{sec:refine}

Earlier sections gave noise-tolerant algorithms for the dictionary learning problem with sample complexity $O(\poly(n,m,k)/\epsilon^2)$. This dependency on $\epsilon$ is necessary for any noise-tolerant algorithm since even if the dictionary has only one vector, we need $O(1/\epsilon^2)$ samples to estimate the vector in presence of noise. However when $Y$ is exactly equal to $AX$ we can hope to recover the dictionary with better running time and much fewer samples. In particular, \cite{GWW} recently established that $\ell_1$-minimization is locally correct for incoherent dictionaries, therefore it seems plausible that given a very good estimate for $A$ there is some algorithm that computes a refined estimate of $A$ whose running time and sample complexity have a better dependence on $\epsilon$. 

In this section we analyze the local-convergence of an algorithm that is similar to K-SVD~\citep{AEB}; see Algorithm~\ref{alg:iteraverage} {\sc IterativeAverage}. Recall $B_S$ denotes the submatrix of $B$ whose columns are indices in $S$; also, $P^+ = (P^TP)^{-1}P^T$ is the left-pseudoinverse of the matrix $P$. Hence $P^+P = I$, $PP^+$ is the projection matrix to the span of columns of $P$.

\begin{fragment*}[t]
\caption{
\label{alg:iteraverage}{\sc IterativeAverage}, \textbf{Input:} Initial estimation $B$, $\|B_i - A_i\| \le \epsilon$, $q$ samples (independent of $B$) $Y^{(1)}, Y^{(2)}, ... Y^{(q)}$ \vspace*{0.01in}
}

\begin{enumerate} \itemsep 0pt
\small 
\item For each sample $i$, let $\Omega^{(i)} = \{j: |\langle Y^{(i)},B_j\rangle| > 1/2\}$
\item For each dictionary element $j$
\item $\quad$ Let $\mathcal{C}_j^+$ be the set of samples that have inner product more than $1/2$ with $B^{(j)}$ ($\mathcal{C}_j^+ =\{i:\langle Y^{(i)}, B_j\rangle > 1/2\}$)
\item $\quad$ For each sample $i$ in $\mathcal{C}_j^+$
\item $\quad$ $\quad$ Let $\hat{X}^{(i)} = B_{\Omega^{(i)}}^+ Y^{(i)}$
\item $\quad$ $\quad$ Let $Q_{i,j} = Y^{(i)} - \sum_{t\in \Omega^{(i)}\backslash\{j\}} B_t \hat{X}^{(i)}_t $
\item $\quad$ Let $B'_j = \sum_{i\in\mathcal{C}_j^+} Q_{i,j}/\|\sum_{i\in\mathcal{C}_j^+} Q_{i,j}\|$.
\item Output $B'$.
\end{enumerate} 

\end{fragment*}

The key lemma of this section shows the error decreases by a constant factor in each round of {\sc IterativeAverage} (provided that it was suitably small to begin with). Let  $\epsilon_0 \leq 1/100k$. 

\begin{theorem}\label{thm:inter}
Suppose the dictionary $A$ is $\mu$-incoherent with $\mu/\sqrt{n} < 1/k\log k$, initial solution is $\epsilon < \epsilon_0$ close to the true solution (i.e. for all $i$ $\|B_i-A_i\| \le \epsilon$). With high probability the output of {\sc IterativeAverage} is a dictionary $B'$ that satisfies $\|B'_i - A_i\| \le (1 - \delta) \epsilon$, where $\delta$ is a universal positive constant. Moreover, the algorithm runs in time $O(qnk^2)$ and succeeds with high probability when number of samples $q = \Omega(m\log^2 m)$. 
\end{theorem}

\noindent We will analyze the update made to the first column $B_1$, and the same argument will work for all columns (and hence we can apply a union bound to complete the proof). To simplify the proof, we will let $\xi$ denote arbitrarily small constants (whose precise value will change from line to line). First, we establish some basic claims that will be the basis for our analysis of {\sc IterativeAverage}. 


\begin{claim}\label{claim:signset}
Suppose $A$ is a $\mu$ incoherent matrix with $\mu/\sqrt{n} < 1/k\log k$. If for all $i$, $\|B_i - A_i\| \le \epsilon_0$ then {\sc IterativeAverage} recovers the correct support for each sample (i.e. $\Omega^{(i)} = \mbox{supp}(X^{(i)})$) and the correct sign (i.e. $\mathcal{C}^+_j = \{j:X^{(i)}_j > 0\}$) \footnote{Notice that this is not a ``with high probability'' statement, the support is {\em always} correctly recovered. That is why we use $\Omega^{(i)}$ both in the algorithm and for the true support}
\end{claim}

\begin{proof}
We can compute $\langle Y^{(i)}, B_1 \rangle = \sum_{j \in \Omega^{(i)}} X^{(i)}_j \langle A_j, B_1 \rangle$ and the total contribution of all of the terms besides $X^{(i)}_1 \langle A_1, B_1 \rangle$ for $j \neq 1$ is at most $1/3$. This implies the claim. 
\end{proof}

\begin{claim}\label{claim:incoherent}
The set of columns $\{B_i\}_i$ is $\mu' = \mu+O(k/\sqrt{n})$-incoherent where $\mu'/\sqrt{n} \le 1/10k$.
\end{claim}

To simplify the notation, let us permute the samples so that $\mathcal{C}^+_1 = \{1,2,...,l\}$. The probability that $X^{(i)}_1 > 0$ is $\Theta(k/m)$ and so for $q = \Theta(m \log^2 m)$ samples with high probability the number of samples $l$ where $X^{(i)}_1 > 0$ is $\Omega(qk/m) = \Omega(k \log^2 m)$. 

\begin{definition}
Let $M_i$ be the matrix $(0,B_{\Omega^{(i)}\backslash \{1\}})B_{\Omega^{(i)}}^+$.
\end{definition}

\noindent Then we can write $Q_{i,1} = (I-M_i) Y^{(i)}$. Let us establish some basic properties of $M_i$ that we will need in our analysis:

\begin{claim}
$M_i$ has the following properties: (1) $M_i B_1 = 0$ (2) For all $j\in \Omega^{(i)}\backslash\{1\}$, $M_i B_j = B_j$ and (3) $\|M_i\| \le 1 + \xi$
\end{claim}

\begin{proof}
The first and second property follow immediately from the definition of $M_i$, and the third property follows from the Gershgorin disk theorem. 
\end{proof}

For the time being, we will consider the vector $\hat{B}_1 = \sum_{i=1}^l Q_{i,1} / \sum_{i=1}^l X^{(i)}_1$. We cannot compute this vector directly (note that $\hat{B}_1$ and $B'_1$ are in general different) but first we will show that $\hat{B}_1$ and $A_1$ are suitably close. To accomplish this, we will first find a convenient expression for the error:

\begin{lemma}\label{lemma:noisypower}
\begin{equation}
A_1 - \hat{B}_1 = \sum_{i=1}^l \frac{X^{(i)}_1}{\sum_{i=1}^l X^{(i)}_1}M_i(A_1 - B_1) - \frac{\sum_{i=1}^l \sum_{j\in \Omega^{(i)}\backslash\{1\}} (I-M_i) (A_j-B_j) X^{(i)}_j}{\sum_{i=1}^l X^{(i)}_1}.\label{eqn:iter}
\end{equation}
\end{lemma}

\begin{proof}
The proof is mostly carefully reorganizing terms and using properties of $M_i$'s to simplify the expression.

Let us first compute $\hat{B}_1 - B_1$:
\begin{align*}
\hat{B}_1 - B_1 & = \frac{\sum_{i=1}^l X^{(i)}_1((I-M_i)A_1 - B_1) + \sum_{i=1}^l \sum_{j\in \Omega^{(i)}\backslash\{1\}} (I-M_i) A_j X^{(i)}_j}{\sum_{i=1}^l X^{(i)}_1}\\
& = \sum_{i=1}^l \frac{ X^{(i)}_1}{\sum_{i=1}^l X^{(i)}_1}(I-M_i)(A_1 - B_1) + \frac{\sum_{i=1}^l \sum_{j\in \Omega^{(i)}\backslash\{1\}} (I-M_i) (A_j-B_j) X^{(i)}_j}{\sum_{i=1}^l X^{(i)}_1}.
\end{align*}
The last equality uses the first and second properties of $M_i$ from the above claim. Consequently we have
\begin{align*}
A_1 - \hat{B}_1 & = (A_1 - B_1) - (\hat{B}_1 - B_1) \\
& = \sum_{i=1}^l \frac{X^{(i)}_1}{\sum_{i=1}^l X^{(i)}_1}M_i(A_1 - B_1) - \frac{\sum_{i=1}^l \sum_{j\in \Omega^{(i)}\backslash\{1\}} (I-M_i) (A_j-B_j) X^{(i)}_j}{\sum_{i=1}^l X^{(i)}_1}.
\end{align*}
And this is our desired expression.
\end{proof}

We will analyze the two terms in the above equation separately. The second term is the most straightforward to bound, since it is the sum of independent vector-valued random variables (after we condition on the support $\Omega^{(i)}$ of each sample in $\mathcal{C}^+_1$. 


\begin{claim}\label{claim:vectorb}
If $l > \Omega(k\log^2 m)$, then with high probability the second term of Equation (\ref{eqn:iter}) is bounded by $\epsilon/100$.
\end{claim}

\begin{proof}
The denominator is at least $l$ and the numeratoris the sum of at most $lk$ independent random vectors with mean zero, and whose length is at most $3C\epsilon$. We can invoke the vector Bernstein's inequality~\cite{G}, and conclude that the sum is bounded by $O(C\sqrt{lk}\log m \epsilon)$ with high probability. After normalization the second term is bounded by $\epsilon / 100$.
\end{proof}

All that remains is to bound the first term. Note that the coefficient of $\|M_i (A_1-B_1)\|$ is independent of the support, and so the first term will converge to its expectation \-- namely $\E[\|M_i (A_1-B_1)\|]$. So it suffices to bound this expectation. 


\begin{lemma}
$\E[\|M_i (A_1-B_1)\|] \le (1 - \delta) \epsilon.$
\end{lemma}

\begin{proof}
We will break up $A_1-B_1$ onto its component ($x$) in the direction $B_1$ and its orthogonal component ($y$) in $B_1^\perp$. First we bound the norm of $x$: $$\|x\| = |\langle A_1-B_1,B_1\rangle | = |\langle A_1-B_1,A_1-B_1\rangle|/2 \le \epsilon^2/2$$
Next we consider the component $y$. Consider the supports $\Omega^{(1)}$ and $\Omega^{(2)}$ of two random samples from $\mathcal{C}^+_1$. These sets certainly intersect at least once, since both contain $\{1\}$. Yet with probability at least $2/3$ this is their only intersection (e.g. see Claim~\ref{claim:intersect}). If so, let $S = (\Omega^{(1)}\cup \Omega^{(2)})\backslash \{1\}$. Recall that $\|B_S^T \| \leq 1 + \xi$. However $B_S^Ty$ is the concatenation of $B_{\Omega^{(1)}}^T y$ and $B_{\Omega^{(2)}}^Ty$ and so we conclude that $\|B_{\Omega^{(1)}}^T y\| + \|B_{\Omega^{(2)}}^Ty\| \leq (1+ \xi)\sqrt{2}$. Since the spectral norm of $(0,B_{\Omega^{(i)}\backslash \{1\}})$ is bounded, we conclude that $\|M_1 y\| + \|M_2y\| \leq (1+ \xi)\sqrt{2}$. This implies that $$\E[\|M_i (A_1-B_1)\|] \leq \E[\|M_i x\|] + \E[\|M_i y\|] \leq (2/3) (1+ \xi)(\sqrt{2}/2)\epsilon + (1/3) (1+ \xi)\epsilon + \epsilon^2/2$$
And this is indeed at most $(1-\delta)\epsilon$ which concludes the proof of the lemma.
\end{proof}

Combining the two claims, we know that with high probability $\hat{B}_1$ has distance at most $(1-\delta)\epsilon$ to $A_1$. However, $B'_1$ is not equal to $\hat{B}_1$ (and we cannot compute $\hat{B}_1$ because we do not know the normalization factor). The key observation here is $\hat{B}_1$ is a multiple of $B'_1$, the vector $B'_1$ and $A_1$ all have unit norm, so if $\hat{B}_1$ is close to $A_1$ the vector $B'_1$ must also be close to $A_1$.

\begin{claim}\label{claim:vectors}
If $x$ and $y$ are unit vectors, and $x'$ is a multiple of $x$ then $\|x' - y\| \leq \epsilon < 1$ implies that $\|x - y\| \leq  \epsilon \sqrt{1 + \epsilon^2}$
\end{claim}

\begin{proof}
We have that 
$\|x - y\|^2 = \sin^2 \theta + (1 - \cos \theta)^2$ where $\theta$ is the angle between $x$ and $y$. Note that $\sin \theta \leq \|x' - y\| \leq \epsilon$ so hence
$\|x - y\| \leq \sqrt{ \epsilon^2 + (1 - \sqrt{1-\epsilon^2})^2}$. Note that for $ 0 \leq a \leq 1$ we have $1 - a \leq \sqrt{1 - a}$ and this implies the claim. 
\end{proof}


This concludes the proof of Theorem~\ref{thm:inter}. To bound the running time, observe that for each sample, the main computations involve computing the pseudo-inverse of a $n\times k$ matrix, which takes $O(nk^2)$ time.

\section{A Higher Order Algorithm}\label{asec:higher}

\label{sec:higher}

\begin{fragment*}[t]
\caption{
\label{alg:ocluster2}{\sc OverlappingCluster2}, \textbf{Input: } $p$ samples $Y^{(1)}, Y^{(2)}, ..., Y^{(p)}$, integer $\ell$ \vspace*{0.01in}
}

\begin{enumerate} \itemsep 0pt
\small 
\item Compute a graph $G$ on $p$ nodes where there is an edge between $i$ and $j$ iff $|\langle Y^{(i)}, Y^{(j)} \rangle | > 1/2$
\item Set $T = \frac{pk}{Cm2^{\ell}}$
\item Repeat $\Omega(k^{\ell-2} m \log^2 m)$ times:
\item $\qquad$ Choose a random node $u$ in $G$, and $\ell-1$ neighbors $u_1, u_2, ... u_{\ell-1}$
\item $\qquad$ If $|\Gamma_G(u) \cap \Gamma_G(u_1) \cap ... \cap \Gamma_G(u_{\ell-1})| \geq T $
\item $\qquad$ $\qquad$  Set $S_{u_1, u_2, ... u_{\ell-1}} = \{ w \mbox{ s.t. } |\Gamma_G(u) \cap \Gamma_G(u_1) \cap ... \cap \Gamma_G(w)| \geq T \} \cup \{u_1, u_2, ... u_{\ell-1}\}$
\item Delete any set $S_{u_1, u_2, ... u_{\ell-1}}$ if $u_1, u_2, ... u_{\ell-1}$ are contained in a strictly smaller set $S_{v_1, v_2, ... v_{\ell-1}}$
\item Output the remaining sets $S_{u_1, u_2, ... u_{\ell-1}}$
\end{enumerate} 

\end{fragment*}

Here we extend the algorithm {\sc OverlappingCluster} presented in Section~\ref{sec:overlap} to succeed even when $k \leq c \min(m^{1/2 - \eta},$ $\sqrt{n}/\mu \log n)$. The premise of {\sc OverlappingCluster} is that we can distinguish whether or not a triple of samples $Y^{(1)},Y^{(2)}, Y^{(3)}$ has a common intersection based on their number of common neighbors in the connection graph. However for $k = \omega(m^{2/5})$ this is no longer true! But we will instead consider higher-order groups of sets. In particular, for any $\eta > 0$ there is an $\ell$ so that we can distinguish whether an $\ell$-tuple of samples $Y^{(1)}, Y^{(2)}, ..., Y^{(\ell)}$ has a common intersection or not based on their number of common neighbors, and this test succeeds even for $k = \Omega(m^{1/2 - \eta})$. 

The main technical challenge is in showing that if the sets $\Omega^{(1)}, \Omega^{(2)}, ..., \Omega^{(\ell)}$ do not have a common intersection, that we can upper bound the probability that a random set $\Omega$ intersects each of them. To accomplish this, we will need to bound the number of ways of piercing $\ell$ sets $\Omega^{(1)}, \Omega^{(2)}, ..., \Omega^{(\ell)}$ that have bounded pairwise intersections by at most $s$ points (see definitions below and Lemma~\ref{lemma:pierce}), and this is the key to analyzing our higher order algorithm {\sc OverlappingCluster2}. We will defer the proofs of the key lemmas and the description of the algorithm in this section to Appendix~\ref{asec:higher}. 

Nevertheless what we need is an analogue of Claim~\ref{claim:upper} and Lemma~\ref{lemma:lower}. The first is easy, but what about an analogue of Lemma~\ref{lemma:lower}? To analyze the probability that a set $\Omega$ intersects each of the sets $\Omega^{(1)}, \Omega^{(2)}, ..., \Omega^{(\ell)}$ we will rely on the following standard definition:

\begin{definition}
Given a collection of sets $\Omega^{(1)}, \Omega^{(2)}, ..., \Omega^{(\ell)}$, the {\em piercing number} is the minimum number of points $p_1, p_2, ..., p_r$ so that each set contains at least one point $p_i$. 
\end{definition}

The notion of piercing number is well-studied in combinatorics (see e.g. \cite{Mat}). However, one is usually interested in upper-bounding the piercing number. For example, a classic result of Alon and Kleitman concerns the $(p,q)$-problem \citep{AK}: Suppose we are given a collection of sets that has the property that each choice of $p$ of them has a subset of $q$ which intersect. Then how large can the piercing number be? Alon and Kleitman proved that the piercing number is at most a fixed constant $c(p,q)$ independent of the number of sets \citep{AK}. 

However, here our interest in piercing number is not in bounding the minimum number of points needed but rather in analyzing how many ways there are of piercing a collection of sets with at most $s$ points, since this will directly yield bounds on the probability that $\Omega$ intersects each of $\Omega^{(1)}, \Omega^{(2)}, ..., \Omega^{(\ell)}$. We will need as a condition that each pair of sets has bounded intersection, and this holds in our model with high-probability. 

\begin{claim}
With high probability, the intersection of any pair $\Omega^{(1)}, \Omega^{(2)}$ has size at most $Q$
\end{claim}

\begin{definition}
We will call a set of $\ell$ sets a $(k, Q)$ family if each set has size at most $k$ and  the intersection of each pair of sets has size at most $Q$.
\end{definition}

\begin{lemma}\label{lemma:pierce}
The number of ways of piercing $(k, Q)$ family (of $\ell$ sets) with $s$ points is at most $(\ell k)^s$. And crucially if $\ell \geq s + 1$, then the number of ways of piercing it with $s$ points is at most $Qs(s+1)(\ell k)^{s-1}$.
\end{lemma}

\begin{proof}
The first part of the lemma is the obvious upper bound. Now let us assume $\ell \geq s + 1$: Then given a set of $s$ points that pierce the sets, we can partition the $\ell$ sets into $s$ sets based on which of the $s$ points is hits the set. (In general, a set may be hit by more than one point, but we can break ties arbitrarily). Let us fix any $s+1$ of the $\ell$ sets, and let $U$ be the the union of the pairwise intersections of each of these sets. Then $U$ has size at most $Q s(s+1)$. Furthermore by the Pigeon Hole Principle, there must be a pair of these sets that is hit by the same point. Hence one of the $s$ points must belong to the set $U$, and we can remove this point and appeal to the first part of the lemma (removing any sets that are hit by this point). This concludes the proof of the second part of the lemma, too.
\end{proof}

\begin{theorem}\label{thm:ocluster2}
The algorithm {\sc OverlappingCluster2}$(\ell)$ finds an overlapping clustering where each set corresponds to some $i$ and contains all $Y^{(j)}$ for which $X^{(j)}_i \neq 0$. The algorithm runs in time $\widetilde{O}(k^{\ell-2}m p + p^2 n)$ and succeeds with high probability if $k \leq c \min(m^{(\ell -1)/(2\ell -1)}, \frac{\sqrt{n}}{\mu \log n})$ and if $p =  \Omega(m^2/k^2 \log m + k^{\ell-2} m \log^2 m)$ 
\end{theorem}

In order to prove this theorem we first give an analogue of Claim~\ref{claim:upper}:

\begin{claim}\label{claim:upper2}
Suppose $\Omega^{(1)} \cap \Omega^{(2)} \cap ... \cap \Omega^{(\ell)} \neq \emptyset$, then $$Pr_Y[ \mbox{ for all } j = 1, 2, ..., \ell, |\langle Y, Y^{(j)} \rangle | > 1/2]  \geq \frac{k}{2m}$$
\end{claim}

\noindent The proof of this claim is identical to the proof of Claim~\ref{claim:upper}. Next we give the crucial corollary of Lemma~\ref{lemma:pierce}. 

\begin{corollary}
The probability that $\Omega$ hits each set in a $(k, Q)$ family (of $\ell$ sets) is at most $$  \sum_{2 \leq s \leq \ell -1} (Qs(s+1) (\ell k)^{s-1}) \Big (\frac{k}{m} \Big )^s + \sum_{s \geq \ell} \Big (\frac{\ell k^2}{m} \Big )^s $$ where $C_s$ is a constant depending polynomially on $s$. 
\end{corollary}

\begin{proof}
We can break up the probability of the event that $\Omega$ hits each set in a $(k, Q)$ family into another family of events. Let us consider the probability that $X$ pierces the family with $s \leq \ell -1$ points or $s \geq \ell$ points. In the former case, we can invoke the second part of Lemma~\ref{lemma:pierce} and the probability that $X$ hits any particular set of $s$ points is at most $(k/m)^s$. In the latter case, we can invoke the first part of Lemma~\ref{lemma:pierce}. 
\end{proof}

Note that if $k \leq m^{1/2}$ then $k/m$ is always greater than or equal to $k^{s-1} (k/m)^s$. And so asymptotically the largest term in the above sum is $(k^2/m)^\ell$ which we want to be asymptotically smaller than $k/m$ which is the probability in Claim~\ref{claim:upper2}. So if $k \leq c m^{(\ell -1)/(2\ell -1)}$ then above bound is $o(k/m)$ which is asymptotically smaller than the probability that a given set of $\ell$ nodes that have a common intersection are each connected to a random (new) node in the connection graph. So again, we can distinguish between whether or not an $\ell$-tuple has a common intersection or not and this immediately yields a new overlapping clustering algorithm that works for $k$ almost as large as $ \sqrt{m}$, although the running time depends on how close $k$ is to this bound.



\section{Discussion}
\label{sec:implement}
This paper shows it is possible to provably get around the chicken-and-egg problem inherent in dictionary
learning: not knowing $A$ seems to prevents recovering $X$'s and vice versa.
By using combinatorial techniques to recover the support of each $X$ without knowing the dictionary, our algorithm suggests a new way to design algorithms. 

Currently the running time is $\widetilde{O}(p^2n)$ time, which may be too slow for large-scale problems. But our algorithm suggests more heuristic versions of recovering the support
that are more efficient.  One alternative is to construct the connection graph $G$ and then find the overlapping clustering by running a truncated power method~\citep{YZ} on $e_i+e_j$ (a vector that is one on indices $i$, $j$ and zero elsewhere and $(i, j)$ is an edge). In experiments, this recovers a good enough approximation to the true clustering that can then be used to smartly initialize KSVD so that it does not have to start from scratch. In practice, this yields a hybrid method that converges much more quickly and succeeds more often. Thus we feel that in practice the best algorithm may use
algorithmic ideas presented here.

We note that for dictionary learning, making stochastic assumptions seems unavoidable. Interestingly, our experiments help to corroborate some of the assumptions. For instance, the condition $\E[X_i|X_i\ne 0] = 0$ used in our best analysis also seems necessary for KSVD; empirically we have seen its performance degrade
when this is violated.

\section*{Acknowledgements} We would like to thank Aditya Bhaskara, Tengyu Ma and Sushant Sachdeva for numerous helpful
discussions throughout various stages of this work.

\newpage

\bibliographystyle{plain}
\bibliography{dictionary}

\newpage

\appendix

\section{Clustering Using Only Bounded 3-wise Moment}\label{asec:cc}
\label{asec:boundedmoment}
When the support of $X$ has only bounded 3-wise moment, it is possible to have two supports $\Omega$ with large intersection. In that case checking the number of common neighbors cannot correctly identify whether the three samples have a common intersection. In particular, there might be false positives (three samples with no common intersection but has many common neighbors) but no false negatives (still all samples with common intersection will have many common neighbors). The algorithm can still work in this case, because it is unlikely for the two supports to have a very large intersection:

\begin{lemma} \label{lemma:intersection}
Suppose $\Gamma$ has bounded 3-wise moments, $k = cm^{2/5}$ for some small enough constant $c>0$. For any set $\Omega$ of size $k$, the probability that a random support $\Omega'$ from $\Gamma$ has intersection larger than $m^{1/5}/100$ with $\Omega$ is at most $O(m^{-6/5})$.
\end{lemma}

\begin{proof}
Let $T$ be the number of triples in the intersection of $\Omega$ and $\Omega'$. For any triple in $\Omega$, the probability that it is also in $\Omega'$ is at most $O(k^3/m^3)$ by bounded 3-wise moment. Therefore $\E[T] \le {k\choose 3} O(k^3/m^3) = O(k^6/m^3)$.

On the other hand, whenever $\Omega$ and $\Omega'$ has more than $m^{1/5}/100$ intersections, $T$ is larger than ${m^{1/5}/100 \choose 3}$. By Markov's inequality we know $\Pr[|\Omega\cap \Omega'|\ge m^{1/5}/100] \le O(m^{-6/5})$.
\end{proof}

Since the probability of having false positives is small (but not negligible), we can do a simple trimming operation when we are computing the set $S_{u,v}$ in Algorithm~\ref{alg:ocluster}. We shall change the definition of $S_{u,v}$ as follows:
\begin{enumerate}
\item Set $S'_{u,v} = \{w:|\Gamma_G(u)\cap\Gamma_G(v)\cap\Gamma_G(w)|\ge T\} \cup \{u,v\}$.
\item Set $S_{u,v} = \{w: w\in S'_{u,v} \mbox{ and }|\Gamma_G(w)\cap S'_{u,v}| \ge T\}$.
\end{enumerate}

Now $S'_{u,v}$ is the same as the old definition and may have false positives. However, intuitively the false positives are not in the cluster so they cannot have many connections to the cluster, and will be filtered out in the second step. In particular, we have the following lemma:

\begin{lemma}\label{lemma:identifyingpair}
If $(u,v)$ is an indentifying pair (as defined in Definition~\ref{def:identifyingpair}) for $i$, then with high probability $S_{u,v}$ is the set $\mathcal{C}_i = \{j:i\in \Omega^{(j)}\}$.
\end{lemma}

\begin{proof}
First we argue the set $S'_{u,v}$ is the union of $\mathcal{C}_i$ with a small set. By Claim~\ref{claim:upper} and Chernoff bound, for all $w\in \mathcal{C}_i$ $u,v,w$ has more than $T$ common neighbors, so $w\in S'_{u,v}$. On the other hand, if $w\not\in \mathcal{C}_i$ but $w\in S'_{u,v}$, then by Lemma~\ref{lemma:lower} we know $\Omega^{(w)}$ must have a large intersection with either $\Omega^{(u)}$ or $\Omega^{(v)}$, which has probability only $O(m^{-6/5})$ by Lemma~\ref{lemma:intersection}. Therefore again by concentration bounds with high probability $|S'_{u,v}\backslash\mathcal{C}_i| \le p/m \ll T$.

Now consider the second step. For the samples in $\mathcal{C}_i$, the probability that they are connected to another random sample in $\mathcal{C}_i$ is $1-O(k^2/m)$, so by concentration bounds with high probability they have at least $T$ neighbors in $\mathcal{C}_i$, and they will not be filtered and are still in $S_{u,v}$. On the other hand, for any vertex $w\not\in \mathcal{C}_i$, the expected number of edges from $w$ to $\mathcal{C}_i$ is only $O(k^2/m)|\mathcal{C}_i| \ll T$, and by concentration property, they are concentrated around the expectation with high probability. So for any $w\in S'_{u,v}\backslash \mathcal{C}_i$, it can only have $O(pk^3/m^2)$ edges to $\mathcal{C}_i$, and $O(p/m)$ edges to $S'_{u,v}\backslash \mathcal{C}_i$. The total number of edges to $S'_{u,v}$ is much less than $T$, so all of those vertices are going to be removed, and $S_{u,v} = \mathcal{C}$.
\end{proof}

This lemma ensures after we pick enough random pairs, with high probability all the correct clusters $\mathcal{C}_i$'s are among the $S_{u,v}$'s. There can be ``bad'' sets, but same as before all those sets contains some of the $\mathcal{C}_i$, so will be removed at the end of the algorithm:

\begin{claim}
For any pair $(u,v)$ with $i\in \Omega^{(u)}\cap\Omega^{(v)}$, let $\mathcal{C}_i =\{j:i\in \Omega^{(j)}\}$, then with high probability $\mathcal{C}_i\subseteq S_{u,v}$.
\end{claim}

\begin{proof}
This is essentially in the proof of the previous lemma.
As before by Claim~\ref{claim:upper} we know $\mathcal{C}_i\subseteq S'_{u,v}$. Now for any sample in $\mathcal{C}_i$, the expected number of edges to $\mathcal{C}_i$ is $(1-o(1))|\mathcal{C}_i|$, by concentration bounds we know the number of neighbors is larger than $T$ with high probability. Then we apply union bound for all samples in $\mathcal{C}_i$, and conclude that $\mathcal{C}_i\subseteq  S_{u,v}$.
\end{proof}

\section{Extensions: Proof Sketch of Theorem~\ref{thm:weakerassumption}}\label{sec:ext}

Let us first examine how the conditions in the hypothesis of
 Theorem~\ref{thm:mainfront} were used in its proof
and then discuss why they can be relaxed.   

Our algorithm is based on three steps: constructing the connection graph, finding the overlapping clustering, and recovering the dictionary. However if we invoke Lemma~\ref{lemma:graph} (as opposed to Lemma~\ref{lemma:hw}) then the properties we need of the connection graph follow from each $X$ being at most $k$ sparse for $k \leq n^{1/4}/\sqrt{\mu}$ without any distributional assumptions. 

Furthermore, the crucial steps in finding the overlapping clustering are bounds on the probability that a sample $X$ intersects a triple with a common intersection, and the probability that it does so when there is no common intersection (Claim~\ref{claim:upper} and Lemma~\ref{lemma:lower}). Indeed, these bounds hold whenever the probability of two sets intersecting in two or more locations is smaller (by, say, a factor of $k$) than the probability of the sets intersecting once. This can be true even if elements in the sets have significant positive correlation (but for the ease of exposition, we have emphasized the simplest models at the expense of generality). Lastly, Algorithm~\ref{alg:oaverage} we can instead consider the difference between the averages for $S_i$ and $\mathcal{C}_i \backslash S_i$ and this succeeds even if $\E[X_i]$ is non-zero. This last step does use the condition that the variables $X_i$ are independent, but if we instead use Algorithm~\ref{alg:osvd} we can circumvent this assumption and still recover a dictionary that is close to the true one. 

Finally, the ``bounded away from zero'' assumption in Definition~\ref{def:gamma} can be relaxed: the resulting algorithm recovers a dictionary that is close enough to the true one and still allows
sparse recovery. This is because when the distribution has the anti-concentration property, a slight variant of Algorithm~\ref{alg:ocluster} can still find most (instead of all) columns with $X_i \ne 0$.

Using the ideas from this part, we give a proof sketch for Theorem~\ref{thm:weakerassumption}

\begin{proof}[sketch for Theorem~\ref{thm:weakerassumption}]
The proof follows the same steps as the proof of Theorem~\ref{thm:main}. There are a few steps that needs to be modified:

\begin{enumerate}
\item Invoke Lemma~\ref{lemma:graph} instead of Lemma~\ref{lemma:hw}.
\item For Lemma~\ref{lemma:lower}, use the weaker bound on the 4-th moment. This is still OK because $k$ is smaller now.
\item In Definition~\ref{def:Ri}, redefine $R_i^2$ to be $\E_{x\in D^i}[\left<A_i, Ax\right>^2]$.
\item In Lemma~\ref{lemma:var}, use the bound $R_i^2 \alpha^2 + \alpha\sqrt{1-\alpha^2} 2k\sqrt{\mu}/n^{1/4} + (1-\alpha^2)k^2\mu/\sqrt{n}$ in order to take the correlations between $X_i$'s into account.
\end{enumerate}

\end{proof}

\paragraph{Remark:}
Based on different assumptions on the distribution, there are algorithms with different trade-offs. Theorem~\ref{thm:weakerassumption} is only used to illustrate the potential of our approach and does not try to achieve optimal trade-off in every case.

A major difference from class $\Gamma$ is that the $X_i$'s do not have expectation $0$ and are not forbidden from taking values close to $0$ (provided they do have reasonable probability of taking values away from $0$).  Another major difference is that the distribution of $X_i$ {\em can} depend upon 
the values of other nonzero coordinates. The weaker moment condition allows a fair bit of correlation among the set of nonzero coordinates.

\label{rem:anticoncentration}
It is also possible to relax the condition that each nonzero $X_i$ is in $[-C,-1]\cup[1,C]$. Instead we require $X_i$ has magnitude at most  $O(1)$, and has a weak {\em anti-concentration} property: for every $\delta >0$ it has probability at least $c_{\delta} >0$ of exceeding $\delta$ in magnitude. This requires changing Algorithm~\ref{alg:ocluster} in the following ways:

For each set $S$, let $T$ be the subset of vertices that have at least $1-2\delta$ neighbors in $S$: $T =\{i\in S, |\Gamma_G(i)\cap S| \ge (1-2\delta)|S|$. Keep sets $S$ that $1-2\delta$ fraction of the vertices are in $T$ ($|T| \ge (1-2\delta)|S|$).
Here the choice of $\delta$ depend on parameters $\mu, n,k$, and effects the final accuracy of the algorithm. This ensures for any remaining $S$, there must be a single coordinate that every $X^{(i)}$ for $i\in S$ is nonzero on. 

In the last step, only output sets that are significantly different from the previously outputted sets (significantly different means the symmetric difference is at least $pk/5m$)

\section{Discussion: Overlapping Communities}\label{sec:disc}

There is a connection between the approach used here, and the recent work on algorithms for finding overlapping communities (see in particular \cite{AGSS}, \cite{BBBCT}). We can think of the set of samples $Y$ for which $X_i \neq 0$ as a ``community". Then each sample is in more than one community, and indeed for our setting of parameters each sample is contained in $k$ communities. We can think of the main approach of this paper as:

\begin{center}
{\em If we can find all of the overlapping communities, then we can learn an unknown dictionary. }
\end{center}

So how can we find these overlapping communities? The recent papers \cite{AGSS}, \cite{BBBCT} pose deterministic conditions on what constitutes a community (e.g. each node outside of the community has fewer edges into the community than do other members of the community). These papers provide algorithms for finding all of the communities, provided these conditions are met. However for our setting of parameters, both of these algorithms would run in quasi-polynomial time. For example, the parameter ``$d$" in the paper \cite{AGSS} is an upper-bound on how many communities a node can belong to, and the running time of the algorithms in \cite{AGSS} are quasi-polynomial in this parameter. But in our setting, each sample $Y$ belongs to $k$ communities -- one for each non-zero value in $X$ -- and the most interesting setting here is when $k$ is polynomially large. Similarly, the parameter ``$\theta$" in \cite{BBBCT} can be thought of as: If node $u$ is in community $c$, what is the ratio of the edges incident to $u$ that leave the community $c$ compared to the number that stay inside $c$? Again, for our purposes this parameter ``$\theta$" is roughly $k$ and the algorithms in \cite{BBBCT} depend quasi-polynomially on this parameter. 

Hence these algorithms would not suffice for our purposes because when applied to learning an unknown dictionary, their running time would depend quasi-polynomially on the sparsity $k$. In contrast, our algorithms run in polynomial time in all of the parameters, albeit for a more restricted notion of what constitutes a community (but one that seems quite natural from the perspective of dictionary learning). Our algorithm {\sc OverlappingCluster} finds all of the overlapping ``communities" provided that whenever a triple of nodes shares a common community they have many more common neighbors than if they do not all share a single community. The correctness of the algorithm is quite easy to prove, once this condition is met; but here the main work was in showing that our generative model meets these neighborhood conditions.

\end{document}